\documentclass[12pt,draftcls,onecolumn]{IEEEtran}
\usepackage{latexsym}
\usepackage{graphicx}
\usepackage{array}
\usepackage{amsmath}
\usepackage{epsfig}

\newtheorem{thm}{Theorem}
\newtheorem{cor}{Corollary}[thm]

\newtheorem{prop}{Proposition}
\newtheorem{defn}{Definition}
\newtheorem{con}{Conjecture}
\newcommand{\bxi} {\boldsymbol{\xi}}
\newcommand{\btheta} {\boldsymbol{\theta}}
\newcommand{\bx} {\boldsymbol{x}}
\newcommand{\bs} {\boldsymbol{s}}
\newcommand{\br} {\boldsymbol{r}}
\newcommand{\bh} {\boldsymbol{h}}
\newcommand{\bw} {\boldsymbol{w}}
\newcommand{\sH}{\mathcal{H}}
\newcommand{\sC}{\mathcal{C}}

\begin{document}
\title{On Convexity of Error Rates in Digital Communications}

\author{Sergey Loyka, Victoria Kostina, Francois Gagnon

\vspace*{-1\baselineskip}

\thanks{This paper was presented in part at the International Zurich Seminar on Communications (IZS), March 3-5, 2010, Zurich, Switzerland \cite{Loyka-10-Zurich}, and the 2010 IEEE International Symposium on Information Theory, Austin, USA \cite{Loyka-10-Austin}.}
\thanks{S. Loyka is with the School of Electrical Engineering and Computer Science, University of Ottawa, 161 Louis Pasteur, Ottawa, Ontario, Canada, K1N 6N5, E-mail: sergey.loyka@ieee.org.}
\thanks{V. Kostina is with the Department of Electrical Engineering, Princeton University, NJ, 08544, USA (e-mail: vkostina@princeton.edu).}
\thanks{F. Gagnon is with the Department of Electrical Engineering, Ecole de Technologie Supérieure, Montreal QC H3C 1K3, Canada (e-mail: francois. gagnon@etsmtl.ca).}
}

\maketitle

\vspace*{-1\baselineskip}
\begin{abstract}
\vspace*{-0.5\baselineskip}
Convexity properties of error rates of a class of decoders, including the ML/min-distance one as a special case, are studied for arbitrary constellations, bit mapping and coding. Earlier results obtained for the AWGN channel are extended to a wide class of noise densities, including unimodal and  spherically-invariant noise.
Under these broad conditions, symbol and bit error rates are shown to be convex functions of the SNR in the high-SNR regime with an explicitly-determined threshold, which depends only on the constellation dimensionality and minimum distance, thus enabling an application of the powerful tools of convex optimization to such digital communication systems in a rigorous way. It is the decreasing nature of the noise power density around the decision region boundaries that insures the convexity of symbol error rates in the general case. The known high/low SNR bounds of the convexity/concavity regions are tightened and no further improvement is shown to be possible in general.  The high SNR bound fits closely into the channel coding theorem: all codes, including capacity-achieving ones, whose decision regions include the hardened noise spheres (from the noise sphere hardening argument in the channel coding theorem) satisfies this high SNR requirement and thus has convex error rates in both SNR and noise power. We conjecture that \textit{all} capacity-achieving codes have convex error rates. Convexity properties in signal amplitude and noise power are also investigated. Some applications of the results are discussed. In particular, it is shown that \textit{fading is convexity-preserving and is never good in low dimensions under spherically-invariant noise}, which may also include any linear diversity combining.

\end{abstract}

\vspace*{-0.5\baselineskip}
\begin{keywords}
\vspace*{-0.5\baselineskip}
Convexity/concavity, error rate, BER, pairwise probability of error, maximum-likelihood decoding, unimodal noise, spherically-invariant noise.
\end{keywords}

\section{Introduction}
\label{sec:introduction}

Convexity properties play a well-known and important role in optimization problems \cite{Boyd}\cite{Ben-Tal}: it is essentially the class of convex problems that are solvable numerically. Furthermore, significant analytical insights are available for this class, which cannot be said about the general class of nonlinear problems \cite{Ben-Tal}.

In the world of digital communications, various types of error rates often serve  as objective or constraint functions during optimization \cite{Wozencraft}-\cite{Yeh}. Therefore, their convexity properties are of considerable importance. While, in some simple scenarios, the convexity can be established by inspection or differentiation of corresponding closed-form error probability expressions  \cite{Wozencraft}-\cite{Yeh}, this approach is not feasible not only in the general case, but also in most cases of practical importance (e.g. modulation combined with coding etc.), since such expressions are either not known or prohibitively complex \cite{Sason}.

A general approach (i.e. not relying on particular closed-form error probability expressions) to convexity analysis in binary detection problems has been developed in \cite{Azizoglu}. This approach has been later extended to arbitrary multidimensional constellations (which can also include coding) in \cite{Loyka-07}\cite{Loyka-10}. In particular, it has been shown that the symbol error rate (SER) of the maximum-likelihood (ML) decoder operating in the AWGN channel is always convex in SNR in dimensions 1 and 2, and also in higher dimensions at high SNR and concave at low SNR (with explicitly specified boundaries of the high/low SNR regimes), for any modulation and coding. Bit error rate (BER) has been shown to be convex in the high SNR regime as well. These results have been also extended to fading channels. In particular, it was shown that "fading is never good in low dimensions". In a related but independent line of study, a log-concavity property of the SER in SNR [dB] for the multi-dimensional uniform square-grid constellations (M-QAM), including fading and diversity reception, has been established in \cite{Conti} and a number of new local SER bounds have been obtained based on it.

In the present paper, the earlier results in \cite{Loyka-10} are expanded in several directions, including an extension to a class of decoders and a wide class of noise densities (significantly different from Gaussian), as well as tightening the high/low SNR bounds of the convexity/concavity regions reported in \cite{Loyka-10}.

While the utility of the Gaussian noise model is well-known, there are a number of scenarios where it is not adequate, most notably an impulsive noise \cite{Middleton-77}\cite{Middleton-99}\cite{Yao-03}\cite{Conte}-\cite{Haengi} with tails much heavier than Gaussian. To address this, an important and natural generalization of the Gaussian random process has been developed, namely, the spherically-invariant random process (SIRP). It  has found a wide range of applications in communications, information-theoretic and signal processing areas \cite{Vershik}-\cite{Conte-2}. This class of processes, while having some of the important properties of the Gaussian process, significantly extends modeling flexibility and thus can be applied to a wider range of phenomena such as impulsive noise, radar clutter, radio propagation disturbances, bandlimited speech waveforms \cite{Yao-03}\cite{Conte}. While the marginal PDF of a SIRP may be significantly different from Gaussian, this class of processes shares a number of important theoretical properties with the Gaussian process: it is closed under linear transformations, it is the most general class of processes for which the optimal MMSE estimator is linear, and the optimal (ML) decoding is still the minimum distance one (this may also include fading and correlated noise) \cite{Vershik}-\cite{Conte-2}. The present paper will extend this list to include the convexity properties of SER under a SIRP noise, which turn out to be similar to those in the AWGN channel (see section \ref{sec3b} for further details). In addition, a general class of unimodal noise power densities will be considered and conditions on an arbitrary noise density will be formulated under which the SER is convex. In particular, the SER is convex in the SNR provided that the noise power density is decreasing around the decision region boundaries, regardless of its behavior elsewhere. It is convex at high SNR under a unimodal or a SIRP noise, and its is always convex (for any SNR) in low dimensions under SIRP noise. Similar results are obtained for convexity in signal amplitude and noise power (which are important for an equalizer design and a jammer optimization).  All the results formulated for an i.i.d. noise are extended to the case of correlated noise as well.

In general, convexity of the SER does not say anything about convexity of the BER, since the latter depends on pairwise probabilities of error (PEP) and not on the SER \cite{Lassing}\cite{Benedetto}. Since the BER is an important performance indicator and thus appears as an objective in many optimization problems, its convexity properties are also studied here using the generic geometrical framework developed for the SER analysis. The setting is general enough so that the results apply to arbitrary constellations, bit mapping and coding. It turns out that the BER is convex at high SNR for a wide class of noise distributions and a class of decoders, where the high SNR boundary is determined by the constellation minimum distance and dimensionality, all other its details being irrelevant.

While the convexity of the PEP and the BER has been established at high SNR, the question remains: how relevant this high SNR regime is, i.e. does it correspond to realistic(practical) SNR values? This has significant impact on the result's importance and its utility when solving practically-relevant optimization problems. In this paper, we provide a positive answer: the high SNR is almost the same as that required by the channel coding theorem so that any code, including capacity-achieving ones, whose decision regions include the hardened noise spheres (from the sphere packing/hardening arguments in the channel coding theorem \cite{Wozencraft}\cite{Shannon}), are in this range. In other words, the boundary of the high SNR regime is closely matched to that in the channel coding theorem so that arbitrary low probability of error implies its convexity and hence power/time sharing does not help to reduce it further. This complements the well-known result that the capacity cannot be increased by power/time sharing. Any practical code whose decision regions include the hardened noise spheres has also convex SER, PEP and BER. This opens up an opportunity to apply numerous and powerful tools of convex optimization to design of systems with such codes on a rigorous basis.

The main contributions are summarized as follows:

$\bullet$ \ New tighter high/low SNR bounds of the convexity/concavity regions are obtained and it is demonstrated that no further improvement is possible in the general case.

$\bullet$ \ While the earlier results in \cite{Loyka-10} were established for the ML (min-distance) decoders only, the same results are shown to apply to any decoder with center-convex decision regions (see section \ref{sec3a} for details), of which the min-distance one is a special case.

$\bullet$ \ While the earlier results in \cite{Loyka-10} were established for the AWGN channel only, the present paper considers a wide class of noise densities of  which Gaussian is a special case (e.g. generic unimodal, SIRP etc.).  In particular, the SER  and the BER are shown to be convex at high SNR for this wider class as well; the SER turns out to be convex in low dimensions not only for the Gaussian, but also for an arbitrary SIRP noise. The constellation dimensionality and minimum distance appear as the main factors affecting the convexity properties.

$\bullet$ \ The boundary of the high SNR regime (where the SER/BER convexity is ensured) is shown to be closely linked to the channel coding theorem, so that error rates of capacity-achieving codes (with vanishingly-small probability of error) are convex.

$\bullet$ \ Any flat-fading and any linear diversity combining are shown to be convexity-preserving, so that fading is never good in low dimensions under spherically-invariant noise, including linear diversity combining.

Tables 1, 2 and 3 summarize the results for convexity properties in the SNR/signal power, signal amplitude and noise power. Unless otherwise indicated, a non-fading channel, an arbitrary constellation and a decoder with center-convex decision regions are assumed.

\begin{table}[htbp]
\caption{Convexity properties of the SER/PEP/BER in the SNR/signal power.}
\begin{center}
\begin{tabular}{|p{3.5in}|p{2in}|}
\hline
\textbf{Convexity/concavity: scenario} & \textbf{Where} \\
\hline
AWGN: SER is convex at high and concave at low SNR; always convex in low dimensions ($n\le 2)$. & Theorem \ref{thm SER SNR new}, \eqref{eq3a-1}, \eqref{eq3a-2}; Corollary \ref{cor SER equal Omega_i}\\
\hline
Arbitrary noise density: SER is convex if the power density is non-increasing at the boundaries of decision regions & Theorem \ref{thm SER SNR general}, \eqref{eq3a-5}, \eqref{eq3a-6}\\
\hline
Unimodal noise: SER is convex at high SNR; always convex if the noise power density is non-increasing. & Corollary \ref{cor SER SNR unimodal}, \ref{cor SER SNR decreas.}\\
\hline
SIRP noise: SER is convex at high and concave at low SNR; always convex if $n\le 2$.& Theorem \ref{thm SIRP SER convexity}, \eqref{eq3b-3}, \eqref{eq3b-4}\\
\hline
AWGN: BER/PEP are convex at high SNR & See \cite{Loyka-10}, \cite{Loyka-10-Zurich}, \cite{Loyka-10-Austin} and Theorem \ref{thm BER old}\\
\hline
AWGN: BER is convex for capacity-approaching codes & \eqref{eq4-1}-\eqref{eq4-4}; Conjecture 1 \\
\hline
SIRP noise: BER/PEP are convex at high SNR & Theorem \ref{thm BER SIRP}, \eqref{eq4-5}, \eqref{eq4-6}\\
\hline
Fading + SIRP noise (AWGN is a special case): fading is never good in low dimensions, including linear combining  & Propositions \ref{prop fading no good}, \ref{prop SER combining}\\
\hline
Fading channel: any flat-fading and any linear combining are convexity preserving (under any noise) & Proposition \ref{prop.3C.1}, \eqref{eq3c-7}; Propositions \ref{prop SER fading}, \ref{prop SER combining} \\
\hline
\end{tabular}
\label{tab1}
\end{center}
\end{table}

\begin{table}[htbp]
\caption{Convexity properties of the SER in signal amplitude}
\begin{center}
\begin{tabular}{|p{3.5in}|p{2in}|}
\hline
\textbf{Convexity/concavity: scenario} & \textbf{Where} \\
\hline
AWGN: SER in convex at high and concave at low SNR; always convex if $n=1$.& Theorem \ref{thm SER A new}, Corollary \ref{cor SER A new}. \\
\hline
Arbitrary noise density: SER is convex if the noise amplitude density is non-increasing at the boundaries of decision regions & Theorem \ref{thm SER A general}\\
\hline
Unimodal noise: SER is convex at high SNR; always convex if the noise amplitude density is non-increasing. & \eqref{eq3d-2}\\
\hline
SIRP noise: SER is convex at high and concave at low SNR; always convex if $n=1$. & Theorem \ref{thm SIRP conv A}; Corollary \ref{cor T SIRP conv A}. \\
\hline
\end{tabular}
\label{tab2}
\end{center}
\end{table}

\begin{table}[htbp]
\caption{Convexity properties of the SER/PEP/BER in noise power.}
\begin{center}
\begin{tabular}{|p{3.5in}|p{2in}|}
\hline
\textbf{Convexity/concavity: scenario}& \textbf{Where} \\
\hline
AWGN: SER in convex at high(low) and concave at low(high) SNR(noise power).& Theorem \ref{thm SER Pn new}, \eqref{eq5-4}, \eqref{eq5-5}; Corollary \ref{cor SER Pn new}. \\
\hline
SIRP noise: SER is convex at high and concave at low SNR.&
Theorem \ref{thm SER Pn SIRP}.\\
\hline
AWGN: PEP in convex at high and low SNR.& Theorem \ref{thm PEP/BER conv Pn}, \eqref{eq5-10}, \eqref{eq5-11}.\\
\hline
AWGN: BER in convex at high SNR.& Corollary \ref{cor BER conv Pn}, \eqref{eq5-12}. \\
\hline
SIRP noise: PEP/BER are convex at high SNR & Corollary \ref{cor PEP/BER Pn SIRP}. \\
\hline
\end{tabular}
\label{tab3}
\end{center}
\end{table}

\section{System Model}

The standard baseband discrete-time system model in an additive noise channel, which includes matched filtering and sampling, is
\begin{equation}
\label{eq2-1}
{\rm {\bf r}}={\rm {\bf s}}+ \bxi
\end{equation}
where ${\rm {\bf s}}$ and ${\rm {\bf r}}$ are $n$-dimensional vectors representing transmitted and received symbols respectively, ${\rm {\bf s}}\in \left\{ {{\rm {\bf s}}_1 ,{\rm {\bf s}}_2 ,...,{\rm {\bf s}}_M } \right\}$, a set of $M$ constellation points, $\bxi$ is an additive white noise. Several noise models will be considered, including the AWGN one, in which case $\bxi \sim {\cal N}({\rm {\bf 0}},\sigma _0^2 {\rm {\bf I}})$, and the corresponding probability density function (PDF) is
\begin{equation}
\label{eq2-2}
f_\xi ({\bf x})=\left( {2\pi \sigma _0^2 } \right)^{-n/2}e^{-|{\bf x}|^2 / 2\sigma _0^2}
\end{equation}
where $\sigma _0^2 $ is the noise variance per dimension, and $n$ is the constellation dimensionality\footnote{While we consider here a real-valued model, all the results extend to the complex-valued case as well by treating real and imaginary parts as two independent reals, so that $n$-D complex constellation corresponds to $2n$-D real one.}; lower case bold letters denote vectors, bold capitals denote matrices, $x_i$ denotes i-th component of ${\rm {\bf x}}$, $\left| {\rm {\bf x}} \right|$ denotes L$_{2}$ norm of ${\rm {\bf x}}$, $\left| {\rm {\bf x}} \right|=\sqrt {{\rm {\bf x}}^T{\rm {\bf x}}} $, where the superscript $T$ denotes transpose, ${\rm {\bf x}}_i $ denotes i-th vector, $|\bf A|$ denotes the determinant of matrix $\bf A$. The average (over the constellation points) SNR is defined as $\gamma =1/\sigma _0^2 $, which implies the appropriate normalization, $\textstyle{1 \over M}\sum\nolimits_{i=1}^M {\left| {{\rm {\bf s}}_i } \right|^2} =1$, unless indicated otherwise.

More general and distinctly different noise distributions will be considered as well, which include the SIRP and unimodal noise, see section \ref{sec3a} and \ref{sec3b} for further details.

In addition to the maximum likelihood decoder (demodulator/detector), which is equivalent to the minimum distance one in the AWGN and some other channels \cite{Goldman}\cite{Conte},
\[
{\rm {\bf \hat {s}}}=\arg \min _{{\rm {\bf s}}_i } \left| {{\rm {\bf r}}-{\rm {\bf s}}_i } \right|,
\]
a general class of decoders with center-convex decision regions (see Definition 1 and Fig. 1) will be considered, for which the min-distance one is a special case. The probability of symbol error $P_{ei} $ (also known as symbol error rate, SER) given that ${\rm {\bf s}}={\rm {\bf s}}_i $ was transmitted is
\begin{equation}
P_{ei} =\Pr \left[ {\left. {{\rm {\bf \hat {s}}}\ne {\rm {\bf s}}_i } \right|{\rm {\bf s}}={\rm {\bf s}}_i } \right]=1-P_{ci}
\end{equation}
where $P_{ci} $ is the probability of correct decision, and the SER averaged over all constellation points is
\begin{equation}
P_e =\sum\nolimits_{i=1}^M {P_{ei} \Pr \left[ {{\rm {\bf s}}={\rm {\bf s}}_i } \right]} =1-P_c
\end{equation}
Clearly, $P_{ei} $ and $P_{ci} $ possess the opposite convexity properties. $P_{ei} $ can be expressed as
\begin{equation}
\label{eq2-5}
P_{ei} =1-\int_{\Omega _i } {f_\xi ({\rm {\bf x}})} d{\rm {\bf x}}
\end{equation}
where $\Omega _i $ is the decision region (Voronoi region)\footnote{${\bf \hat {s}} = {\bf s}_i$ if ${\bf r} \in \Omega _i$. If ${\bf r} \notin \Omega _i \forall i$, an error is declared.  }, and ${\rm {\bf s}}_i $ corresponds to ${\rm {\bf x}}=0$, i.e. the origin is shifted for convenience to the constellation point ${\rm {\bf s}}_i $. For the min-distance decoder, $\Omega _i $ can be expressed as a convex polyhedron \cite{Boyd},
\begin{equation}
\label{eq2-6}
\Omega _i =\left\{ {\left. {\rm {\bf x}} \right|{\rm {\bf Ax}}\le {\rm {\bf b}}} \right\},\mbox{ }{\rm {\bf a}}_j^T =\frac{({\rm {\bf s}}_j -{\rm {\bf s}}_i )}{\left| {{\rm {\bf s}}_j -{\rm {\bf s}}_i } \right|},\mbox{ }b_j =\frac{1}{2}\left| {{\rm {\bf s}}_j -{\rm {\bf s}}_i } \right|
\end{equation}
where ${\rm {\bf a}}_j^T $ denotes j-th row of ${\rm {\bf A}}$, and the inequality in \eqref{eq2-6} is applied component-wise.

Another important performance indicator is the pairwise error probability (PEP) i.e. a probability $\Pr \left\{ {{\rm {\bf s}}_i \to {\rm {\bf s}}_j } \right\}=\Pr \left[ {\left. {{\rm {\bf \hat {s}}}={\rm {\bf s}}_j } \right|{\rm {\bf s}}={\rm {\bf s}}_i } \right]$ to decide in favor of ${\rm {\bf s}}_j $ given that ${\rm {\bf s}}_i $, $i\ne j$, was transmitted, which can be expressed as
\begin{equation}
\Pr \{ {\bf s}_i \to {\bf s}_j \} = \int_{\Omega_j } {f_\xi ({\rm {\bf x}})} d{\rm {\bf x}}
\end{equation}
where $\Omega _j $ is the decision region for ${\rm {\bf s}}_j $ when the reference frame is centered at ${\rm {\bf s}}_i $. The SER can now be expressed as
\begin{equation}
P_{ei} =\sum\nolimits_{j\ne i} {\Pr \left\{ {{\rm {\bf s}}_i \to {\rm {\bf s}}_j } \right\}}
\end{equation}
and the BER can be expressed as a positive linear combination of PEPs \cite{Lassing}
\begin{equation}
\label{eq2-9}
\mbox{BER}=\sum\limits_{i=1}^M {\sum\limits_{j\ne i} {\frac{h_{ij} }{\log _2 M}\Pr \left\{ {{\rm {\bf s}}={\rm {\bf s}}_i } \right\}\Pr \left\{ {{\rm {\bf s}}_i \to {\rm {\bf s}}_j } \right\}} }
\end{equation}
where $h_{ij} $ is the Hamming distance between binary sequences representing ${\rm {\bf s}}_i $ and ${\rm {\bf s}}_j $.

Note that the setup and error rate expressions we are using are general enough to apply to arbitrary multi-dimensional constellations, including coding (codewords are considered as points of an extended constellation). We now proceed to convexity properties of error rates in this general setting.

\section{Convexity of Symbol Error Rates}
\label{sec:convexity}
Convexity properties of symbol error rates of the ML decoder in SNR and noise power have been established in \cite{Loyka-07}\cite{Loyka-10} for arbitrary constellation/coding under ML decoding and AWGN noise and are summarized in Theorem 1 below for completeness and comparison purposes.

\begin{thm}[Theorems 1 and 2 in \cite{Loyka-10}]
\label{thm SER SNR old}
Consider the ML decoder operating in the AWGN channel. Its SER $P_e(\gamma)$ is a convex function of the SNR $\gamma $ for any constellation/coding if $n\le 2$,
\begin{equation}
d^2 P_e(\gamma) / d\gamma^2 = P_e(\gamma)'' \ge 0
\end{equation}
For $n>2$, the following convexity properties hold:

\begin{itemize}
\item $P_{e} $ is convex in the high SNR regime,
\begin{equation}
\gamma \ge (n+\sqrt{2n}) / d_{\min}^2
\end{equation}
where $d_{\min} = \min_i \{d_{\min ,i}\}$ is the minimum distance from a constellation point to the boundary of its decision region over the whole constellation, and $d_{\min ,i}$ is the minimum distance from ${\rm {\bf s}}_i $ to its decision region boundary,

\item $P_{e} $ is concave in the low SNR regime,
\begin{equation}
\gamma \le (n-\sqrt{2n}) / d_{\max}^2
\end{equation}
where $d_{\max} = \max_i \{d_{\max ,i}\}$, and $d_{\max ,i}$ is the maximum distance from ${\rm {\bf s}}_i $ to its decision region boundary,

\item there are an odd number of inflection points, $P_{e}(\gamma)'' = 0$, in the intermediate SNR regime,
\begin{equation}
( n-\sqrt{2n}) / d_{\max}^2 \le \gamma \le ( n+\sqrt{2n} ) / d_{\min}^2
\end{equation}
\end{itemize}
\end{thm}
The same results can be extended to $P_{ei}$ via the substitution $d_{\max(\min)} \rightarrow d_{\max,i(\min,i)}$ in the inequalities above.

\subsection{Convexity in SNR/Signal Power}
\label{sec3a}

Since the high/low SNR bounds in Theorem 1 are only sufficient for the corresponding property, a question arises whether they can be further improved. Theorem 2 gives an affirmative answer and demonstrates that no further improvement is possible.

\begin{thm}
\label{thm SER SNR new}
Consider the ML decoder operating in the AWGN channel. Its SER $P_e(\gamma)$ has the following convexity properties: it is convex in the high SNR regime,
\begin{equation}
\label{eq3a-1}
\gamma \ge (n-2) / d_{\min}^2
\end{equation}
it is concave in the low SNR regime,
\begin{equation}
\label{eq3a-2}
\gamma \le (n-2) / d_{\max}^2
\end{equation}
and there are an odd number of inflection points in-between. The high/low SNR bounds cannot be further improved without further assumptions on the constellation geometry.
\end{thm}

\begin{proof}
See Appendix.
\end{proof}

Note that the high/low SNR bounds in Theorem \ref{thm SER SNR new} are tighter than those in Theorem \ref{thm SER SNR old}, since
\[
n-\sqrt{2n} < n -2 < n+\sqrt{2n} \ \mbox{for} \ n>2.
\]
Convexity of the SER for $n \le 2$ is also obvious from this Theorem. In the case of identical spherical decision regions, a more definite statement can be made.

\begin{cor}
\label{cor SER equal Omega_i}
Consider the case of Theorem \ref{thm SER SNR new} when all decision regions are spheres\footnote{If the received signal does not belong to any of the decision regions, an error is declared.} of the same radius $d$. The following holds:
\begin{itemize}
\item The SER is strictly convex in $\gamma$ in the high SNR regime:
 \[
P_{e}(\gamma)'' > 0 \mbox{ if} \ \gamma > (n-2) /d^2
\]

\item It is strictly concave in the low SNR regime:
 \[
P_{i}(\gamma)'' < 0 \mbox{ if} \ \gamma < (n-2) /d^2
\]

\item There is a single inflection point:
\[
P_{e}(\gamma)'' = 0 \mbox{ if} \ \gamma = (n-2) /d^2
\]
\end{itemize}
\QED
\end{cor}

Note that this result cannot be obtained from Theorem \ref{thm SER SNR old} directly, as the bounds there are not tight. It also follows from this Corollary that the high/low SNR bounds of Theorem \ref{thm SER SNR new} cannot be further improved in general (without further assumptions on the constellation geometry).

The results above are not limited to the AWGN channel but can also be extended to a wide class of noise densities and a class of decoders, as Theorem \ref{thm SER SNR general} below demonstrates. We will need the following definition generalizing the concept of a convex region.

\begin{defn}
A decision region is \textit{center-convex} if any of its points can be "seen" from the center (i.e. the corresponding line segment connecting the point to the center belongs to the region).
\end{defn}

Note that any convex region (e.g. a convex polyhedron) is automatically center-convex but the converse is not necessarily true, so that ML/min-distance decoders are a special case of a generic decoder with center-convex decision regions. As an example, Fig. 1 illustrates such a decision region, which is clearly not convex.

\begin{figure}[t]
\label{fig_1}
\centerline{\includegraphics[width=1.5in]{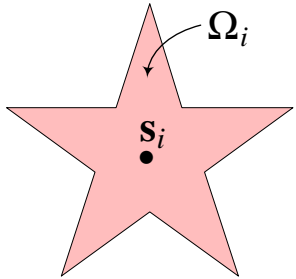}}
\caption{Center-convex decision region $\Omega_i$ centered on $\rm {\bf s}_i$.}
\end{figure}

To generalize the results above to a wide class of noise densities, we  transform the Cartesian noise density $f_{\xi}(\bf{x})$ into the spherical coordinates $(p,\btheta)$,
\begin{eqnarray}
x_1 &=& \sigma_0 \sqrt{p} \sin \theta_1 \\ \notag
x_2 &=& \sigma_0 \sqrt{p} \cos \theta_1 \sin \theta_2 \\ \notag
\vdots \\ \notag
x_{n-1} &=& \sigma_0 \sqrt{p} \cos \theta_1 .. \cos \theta_{n-2} \sin \theta_{n-1} \\ \notag
x_{n} &=& \sigma_0 \sqrt{p} \cos \theta_1 .. \cos \theta_{n-2} \cos \theta_{n-1} \\ \notag
\end{eqnarray}
where $\btheta = \{\theta_1,..,\theta_{n-1}\}$ are the angles, $-\pi/2 \le \theta_i \le \pi/2$ for $i=1...n-2$, $-\pi \le \theta_{n-1} \le \pi$, and $p$ represents the normalized noise instant power $|\bxi|^2/\sigma_0^2$, so that
\begin{equation}
f_{p,\btheta}(p,\btheta) = |\partial {\rm{\bf{x}}} / \partial (p, \btheta)| f_\xi(\bf{x})
\end{equation}
where
\[
J=|\partial {\rm{\bf{x}}} / \partial (p, \btheta)|= \sigma_0^n  p^{n/2-1} \cos^{n-2}\theta_{n-2} \cos^{n-3}\theta_{n-3} .. \cos \theta_1
\]
is the Jacobian of transformation from $\bf{x}$ to $(p, \btheta)$,  and $f_{p,\btheta}(p,\btheta)$ is the noise power density in the spherical coordinates (see \cite{Anderson}\cite{Goldman} for more on spherical coordinates and corresponding transformations). For simplicity of notations, we further drop the subscripts and use $f(p,\btheta)$.

We are now in a position to generalize Theorem \ref{thm SER SNR new} to a wide class of noise densities and the class of center-convex decoders.

\begin{thm}
\label{thm SER SNR general}
Consider a decoder with center-convex decision regions operating in an additive noise channel of arbitrary density $f(p,\btheta)$. The following holds:
\begin{equation}
\label{eq3a-5}
P_e(\gamma)'' \ge 0 \ \mbox{if} \ f_p'(p,\btheta) \le 0 \ \forall \btheta, p \in [\gamma d_{\min}^2, \gamma d_{\max}^2],
\end{equation}
where $f_p'(p,\btheta)=\partial f(p,\btheta)/\partial p$. In particular, $P_e(\gamma)$ is convex in the interval $[\gamma_1, \gamma_2]$ if the noise density $f(p,\btheta)$ is non-increasing in $p$ in the interval $[\gamma_1 d_{\min}^2, \gamma_2 d_{\max}^2]$:
\begin{equation}
\label{eq3a-6}
P_e(\gamma)'' \ge 0 \ \forall \gamma \in [\gamma_1, \gamma_2] \ \mbox{if} \ f_p'(p,\btheta) \le 0 \ \forall \btheta, p \in [\gamma_1 d_{\min}^2, \gamma_2 d_{\max}^2],
\end{equation}
\end{thm}
\begin{proof}
See Appendix.
\end{proof}

Note that it is the (non-increasing) behavior of the noise power density in the annulus $[\gamma_1 d_{\min}^2, \gamma_2 d_{\max}^2]$, i.e. around the boundaries of decision regions, that is responsible for the convexity of $P_e(\gamma)$; the behavior of the noise density elsewhere is irrelevant.

The inequalities in \eqref{eq3a-5} and \eqref{eq3a-6} can be reversed to obtain the corresponding concavity properties. The strict convexity properties can also be established by considering decoders with decision regions of non-zero measure in the corresponding SNR intervals. Convexity of individual SER $P_{ei}$ can be obtained via the substitution $d_{\min(\max)} \rightarrow d_{\min,i(\max,i)}$. It is also straightforward to see that Theorem \ref{thm SER SNR new} is a special case of Theorem \ref{thm SER SNR general}.

Let us now consider more special cases of Theorem \ref{thm SER SNR general}.

\begin{cor}
\label{cor SER SNR unimodal}
Consider a decoder with center-convex decision regions operating in an additive noise channel of a unimodal noise power density\footnote{which is also quasi-concave \cite{Boyd}; many popular probability density functions are unimodal.},
\begin{align}
f_p'(p,\btheta)
\begin{cases}
> 0, \ p < p^* \\
=0, \ p = p^* \\
<0, \ p > p^*
\end{cases}
\end{align}
i.e. it has only one maximum at $p=p^*$; it is an increasing function on one side and decreasing on the other (see e.g Fig. 2, 3). Its SER is convex at high and concave at low SNR:
\begin{align}
\label{eq3a-8}
\begin{cases}
P_e(\gamma)'' > 0, \ \gamma > p^*/d_{\min}^2 \\
P_e(\gamma)'' < 0, \ \gamma < p^*/d_{\max}^2 \\
\end{cases}
\end{align}
\QED
\end{cor}

\begin{cor}
\label{cor SER SNR decreas.}
Consider the case of monotonically-decreasing (in $p$) noise power density, $f_p'(p,\btheta)<0 \ \forall p, \btheta$. Then, the SER is always convex: $P_e(\gamma)'' > 0 \ \forall \gamma$.
\QED
\end{cor}

\begin{figure}[t]
\label{fig_2}
\centerline{\includegraphics[width=3.5in]{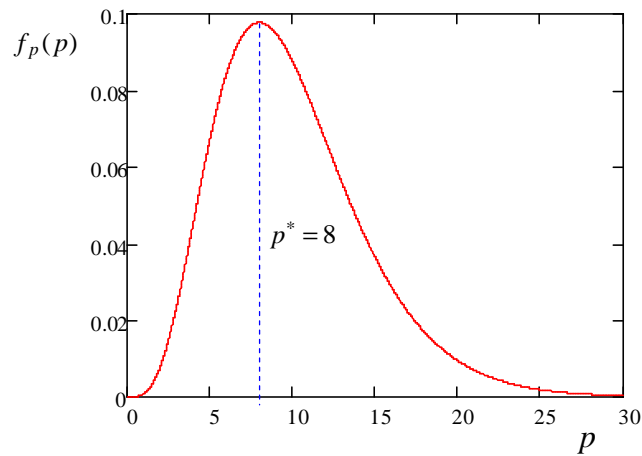}}
\caption{Gaussian noise power density for $n=10$. It is unimodal with $p^* = 8$.}
\end{figure}

\begin{figure}[t]
\label{fig_3}
\centerline{\includegraphics[width=3.5in]{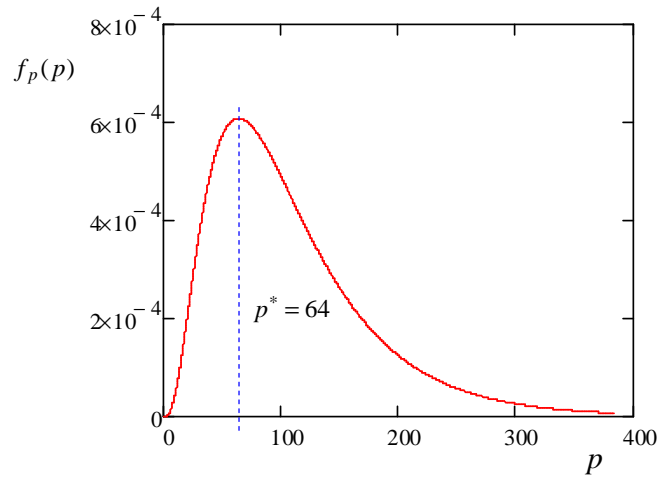}}
\caption{Laplacian noise power density for $n=10$. It is unimodal with $p^* = 64$.}
\end{figure}

\begin{figure}[t]
\label{fig_4}
\centerline{\includegraphics[width=3.5in]{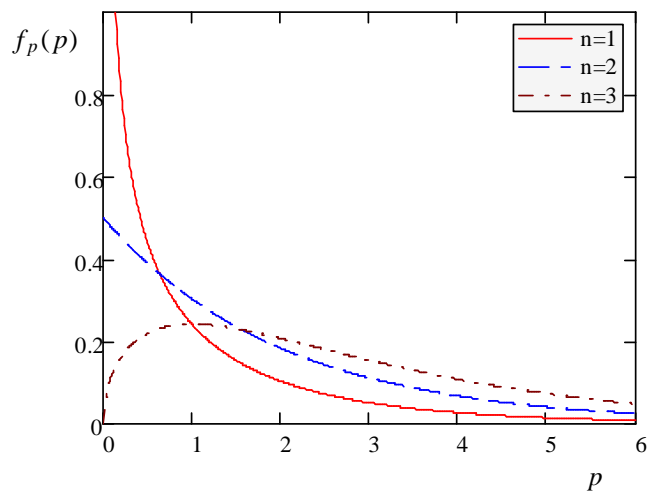}}
\caption{The power density of Gaussian noise: while it is monotonically decreasing for $n=1, 2$ and $\forall p$, it is unimodal for $n \ge 3$.}
\end{figure}

Since the Gaussian noise power density is unimodal with $p^* = n-2$ (see Fig. 2), Corollary \ref{cor SER SNR unimodal} applies to the AWGN channel as well, thereby generalizing Theorem \ref{thm SER SNR new} to decoders with center-convex decision regions. The AWGN for $n=1,2$ is also a special case of Corollary \ref{cor SER SNR decreas.}. These Corollaries allow one to answer the question "Why is the SER in the AWGN channel always convex for $n=1,2$ but not for $n \ge 3$?" -  the reason is the monotonically decreasing (in $p$) nature of the noise power density $f(p,\btheta)$ for any $p$ in the former but not the latter case, see Fig. 4.

Other examples of unimodal densities include Laplacian noise with the Cartesian PDF $f_{\xi}({\bf{x}}) = c \cdot e^{-|\bf{x}|}$, where $c$ is a normalizing constant, so that the spherical one is
\[
f(p,\btheta) = c \cdot p^{n/2-1} e^{-\sqrt{p}} f(\btheta),
\]
where $f(\btheta)$ is the angular density. It is straightforward to see that this power density is unimodal in $p$ with $p^* = (n-2)^2$, see Fig. 3. Importance of this distribution for communication/information theoretic problems is discussed in \cite{Verdu}-\cite{Kotz}. A more general example is a power exponential distribution \cite{Lindsey}-\cite{Gomez}
 \[
f_{\xi}({\bf{x}}) = c \cdot e^{-|{\bf{x}}|^{2\beta}}, \ \beta > 0,
\]
(also known as generalized Gaussian \cite{Conte} or, in a  slightly modified form, as Weibull distribution \cite{Sagias}-\cite{Rinne})  whose spherical density is
\[
f(p,\btheta) = c \cdot p^{n/2-1} e^{-p^\beta} f(\btheta), \beta > 0,
\]
which is also unimodal in $p$ with
\[
p^* = \left( \frac{n-2}{2\beta}\right)^{1/\beta}
\]
This distribution has a heavier (for $\beta < 1$) or lighter (for $\beta >1$) tail than the Gaussian one, so that it offers a significant flexibility in noise modeling. In fact, it was shown in \cite{Gomez} that Weibull distribution can be presented as a mixture of normal distributions, where the variance of normal distribution is treated as a random variable with an $\alpha$-stable distribution. This fits well into a typical model of interference in random wireless networks, where the interference distribution also follows an $\alpha$-stable law \cite{Sousa_Silvester}-\cite{Haengi}: each node transmits a Gaussian (capacity-achieving) signal of a fixed transmit power; at the receiver, the noise power coming from each node is random (due to random distance to transmitting nodes) and follows an $\alpha$-stable law, so that the composite noise instant power follows the power exponential distribution.

Some spherically-invariant random processes or vectors considered in \cite{Yao-73}-\cite{Conte} also belong to the classes considered in these Corollaries or in Theorem \ref{thm SER SNR general}, as discussed next.

\subsection{Convexity of SER under SIRP noise}
\label{sec3b}

In this section, we consider an additive noise channel when the noise distribution follows that of a SIRP. The characterization of the SIRP class is strikingly simple: any SIRP process is conditionally Gaussian, i.e. a Gaussian random process whose variance is a random variable independent of it. In the context of wireless communications, this structure represents such important phenomena as channel fading, random distance between transmitter and receiver etc. Below, we establish the SER convexity properties under a SIRP noise thus generalizing further the results of the previous section.

The following is one of the several equivalent definitions of a SIRP \cite{Vershik}-\cite{Goldman}.

\begin{defn}
A random process $\{X(t), t \in R\}$ is a SIRP if a vector of any of its $n$ samples ${\bf{x}} = \{X(t_1), X(t_2)..X(t_n)\}$ has the PDF of the following form:
\begin{equation}
\label{eq3b-1}
f_{\bf{x}}({\bf{x}}) = c_n h_n({\bf{x}}^T {\bf{C}}_n^{-1} {\bf{x}})
\end{equation}
where ${\bf{C}}_n$ is the covariance matrix, $h_n(r)$ is a non-negative function of the scalar argument $r \ge 0$, and $c_n$ is a normalizing constant. \footnote{An equivalent definition in terms of the characteristic function is also possible. Note also that not any $h_n(r)$ will do the job, but only those satisfying the Kolmogorov consistency condition \cite{Yao-73}-\cite{Goldman}.} \ \QED
\end{defn}

In fact, Definition 2 says that the PDF of SIRP samples depends only on the quadratic form ${\bf{x}}^T {\bf{C}}_n^{-1} {\bf{x}}$ rather than on each entry individually, so that any linear combinations of the entries of ${\bf{x}}$ having the same variance will also have the same PDF \cite{Vershik}. Distributions of the functional form as in \eqref{eq3b-1} are also known as elliptically-contoured distributions \cite{Anderson}. The characterization of SIRP is as follows (the SIRP representation theorem) \cite{Yao-73}-\cite{Goldman}.

\begin{thm}
\label{thm SIRP representation}
A random process is a SIRP iff any set of its samples has a PDF as in \eqref{eq3b-1} with
\begin{equation}
\label{eq3b-2}
h_n(r) = \int_0^{\infty} \sigma^{-n} \exp \left\{-\frac{r}{2\sigma^2}\right\} f(\sigma) d\sigma, \ 0 < r < \infty,
\end{equation}
where $h_n(r)$ is defined by continuity at $r=0$, and $f(\sigma)$ is any univariate PDF. \ \QED
\end{thm}

An equivalent representation is $X(t) = C Y(t)$, where $Y(t)$ is the Gaussian random process of unit variance, and $C$ is an independent random variable of PDF $f(\sigma)$, so that Theorem \ref{thm SIRP representation} basically says that any SIRP can be obtained by modulating the Gaussian random process by an independent random variable \cite{Wise}. A number of PDFs that satisfy Theorem \ref{thm SIRP representation} and corresponding $f(\sigma)$ can be found in \cite{Conte} (which include Laplacian and power exponential densities above).

It was shown in \cite{Conte} that the optimal decoder under the SIRP noise is still the minimum distance one (which follows from the fact that $h_n(r)$ in \eqref{eq3b-2} is monotonically decreasing in $r$). Using this, we are now in a position to establish the SER convexity properties under SIRP noise with $\bf{C} = I$.
\begin{thm}
\label{thm SIRP SER convexity}
Consider an additive SIRP noise channel, where the noise density is as in \eqref{eq3b-1} and \eqref{eq3b-2} with $\bf{C} = I$. Assume that $f(\sigma)$ in \eqref{eq3b-2} has bounded support: $f(\sigma)=0 \ \forall \sigma \notin [\sigma_1, \sigma_2]$. Then, the SER of any decoder with center-convex decision regions (including the min-distance/ML one as a special case) operating in this channel is convex at high and concave at low SNR as follows:
\begin{eqnarray}
\label{eq3b-3}
P_e(p_s)'' &\ge& 0 \ \mbox{if} \ p_s \ge (n-2) \sigma_2^2 / d_{\min}^2 \\
\label{eq3b-4}
P_e(p_s)'' &\le& 0 \ \mbox{if} \ p_s \le (n-2) \sigma_1^2 / d_{\max}^2
\end{eqnarray}
where $p_s$ is the signal power, and $d_{\min(\max)}$ is the minimum (maximum) distance in the normalized constellation (corresponding to $p_s=1$).
\end{thm}
\begin{proof}
See Appendix.
\end{proof}

Note that the high/low SNR bounds are independent of a particular form of $f(\sigma)$, but depend only on the corresponding boundaries of its support set. A particular utility of this Theorem is due to the fact that closed-form expressions of $P_e(p_s)$ are not available in most cases so its convexity cannot be evaluated directly. The following Corollary is immediate.

\begin{cor}
Consider a decoder with center-convex decision regions operating in the SIRP noise channel as in Theorem \ref{thm SIRP SER convexity} without the bounded support assumption. Its SER $P_e(p_s)$ is always convex when $n \le 2$: $P_e(p_s)'' \ge 0 \ \forall p_s$. \ \QED
\end{cor}

Thus, the SER is convex in low dimensions for all the noise densities in Table I in \cite{Conte} (i.e. contaminated normal, generalized Laplace, Cauchy and Gaussian), which extends the corresponding result in Theorem 1 to a generic SIRP noise.

While Corollary 3.1 characterizes the SER convexity for the identical spherical decision regions in the AWGN channel, such a simple characterization is not possible in a SIRP channel in general (when $\sigma_1 \neq \sigma_2$), as Theorem \ref{thm SIRP SER convexity} shows.

\subsection{Non-negative Mixture is Convexity-Preserving}
\label{sec3c}
The next proposition generalizes further the results above and shows that any non-negative mixture of noise densities is convexity-preserving in terms of error rates of a given decoder for any variable of interest. We will need the following definition.
\begin{defn}
Let $\{f_i\}$ be a set of noise densities, $i=1...m$. Its convex hull \cite{Boyd} is any non-negative liniear combination which is also a density,
\[
\sH\{f_i\} = \{f: f = \sum_i \alpha_i f_i, \ \alpha_i \ge 0, \sum_i \alpha_i = 1 \}
\]
\end{defn}

\begin{prop}
\label{prop.3C.1}
Let $P_e[f]$ be an error rate of a given decoder as a functional of noise density $f$ and let all $P_e[f_i]$ be convex, $P_e[f_i]'' \ge 0 \ \forall i$, where the derivative is over any variable of interest (SNR, power/amplitude of signal/noise). Then,
\begin{align}
\label{eq3c-5}
P_e[f]'' \ge 0 \ \forall f \in \sH\{f_i\}
\end{align}
\end{prop}
\begin{proof}
A key observation here is to note, from \eqref{eq2-5}, that that the probability of correct decision $P_c[f] = 1 - P_e[f]$ is a linear functional of the underlying noise density $f$,
\begin{align}
\label{eq3c-6} \notag
P_c\left[\sum_i \alpha_i f_i\right] &= \sum_k \Pr\{\bs=\bs_k\} \int_{\Omega_k} \sum_i \alpha_i f_i(\bx) d\bx \\ \notag
&= \sum_i \alpha_i \sum_k \Pr\{\bs=\bs_k\} \int_{\Omega_k} f_i(\bx) d\bx \\
&= \sum_i \alpha_i P_c\left[f_i\right]
\end{align}
Since $\alpha_i \ge 0$ and each $P_c\left[f_i\right]$ is concave, so is $P_c[f]$, from which the result follows.
\end{proof}

Thus, a convexity of error rates of a given decoder for $\{f_i\}$ is sufficient to insure the convexity for any $f$ in the convex hull of $\{f_i\}$.

The same preservation holds for concavity and also when the corresponding property is strict. It can be further extended to continuous mixtures  as well,
\begin{align}
\label{eq3c-7}
P_e\left[\int \rho(\alpha) f_{\alpha} d\alpha\right]'' \ge 0 \ \mbox{if} \ P_e[f_{\alpha}]'' \ge 0 \ \forall \alpha, \rho(\alpha) \ge 0, \int \rho(\alpha) d\alpha =1
\end{align}
where $f_{\alpha}$ is the noise density parameterized by a continuous mixture parameter $\alpha$ and $\rho(\alpha)$ is a (non-negative) density of $\alpha$.  Such a mixture  can model a fading channel where $\alpha$ represents the channel (random) gain, so that \eqref{eq3c-7} states in fact that (flat) fading is a convexity-preserving process. This will be elaborated in further details in Section \ref{sec 6}.

Observe that the same convexity-preserving property holds in terms of the signal power/amplitude and noise power/amplitude, due to the linearity of $P_c[f]$. Note also that Proposition \ref{prop.3C.1} and \eqref{eq3c-7} extend the confexity/concavity properties to a very broad class of noise densities, which includes, as a special case, the SIRP noise, and do not even assume convex or center-convex decision regions\footnote{To the best of our knowledge, this is the most general known result about the convexity properties of error rates.}.

The following result is a direct consequence of Proposition \ref{prop.3C.1}.

\begin{prop}
\label{prop.3C.2}
Let $\sC$ be a set of all noise densities for which error rates of a given decoder are convex (in any variable of interest). $\sC$ is a convex set.
\end{prop}

\subsection{Convexity in Signal Amplitude}
\label{sec3d}

Convexity of the SER as a function  of signal amplitude $A=\sqrt{\gamma}$, $P_e(A)$, is also important for some optimization problems (e.g. an equalizer design). For the ML decoder operating in the AWGN channel those properties have been established in \cite{Loyka-10}, which are  summarized in Proposition \ref{prop SER A old} for completeness.
\begin{prop}
\label{prop SER A old}
Consider the ML decoder in the AWGN channel. Its SER $P_{ei}(A)$ as a function of signal amplitude $A$ has the following convexity properties:
\begin{itemize}
\item $P_{ei}(A)$ is always convex in $A$ if $n=1$,

\item For $n>1$, it is convex in the large SNR regime $A \ge \sqrt{\alpha_1} /d_{\min,i}$ and concave in the small SNR regime $A \le \sqrt{\alpha_2} /d_{\max,i}$, where
    \[
    \alpha_1 = (2n + 1 +\sqrt{8n+1})/2, \ \alpha_2 = (2n + 1 -\sqrt{8n+1})/2
    \]
    and there are an odd number of inflection points in-between.

\item The same applies to $P_{e}(A)$ via the substitution $d_{\max(\min),i} \rightarrow  d_{\max(\min)}$. \ \QED
\end{itemize}
\end{prop}

The next Theorem provides tighter high/low SNR bounds, which cannot be further improved in general, and also extends the result to any decoder with center-convex decision regions (of which the ML/min-distance one is a special case).

\begin{thm}
\label{thm SER A new}
Consider a decoder with center-convex decision regions operating in the AWGN channel. Its SER $P_{ei}(A)$ as a function of signal amplitude $A$ has the following convexity properties for any $n$:
\begin{itemize}
\item The SER is convex in $A$ in the large SNR regime:
 \[
P_{ei}(A)'' \ge 0 \mbox{ if} \ A \ge \sqrt{n-1} /d_{\min,i}
\]

\item It is concave in the small SNR regime
 \[
P_{ei}(A)'' \le 0 \mbox{ if} \ A \le \sqrt{n-1} /d_{\max,i}
\]

\item There are an odd number of inflection points in-between.

\item The bounds cannot be further tightened in general (without further assumptions on the constellation geometry).

\item The same applies to $P_{e}(A)$ via the substitution $d_{\max(\min),i} \rightarrow  d_{\max(\min)}$.
\end{itemize}
\end{thm}
\begin{proof}
See Appendix.
\end{proof}

Note that the convexity of $P_{ei}(A)$ for $n=1$ and any $A$ follows automatically from this Theorem. It is straightforward to see that the bounds of Theorem \ref{thm SER A new} are indeed tighter than those of Proposition 1, since
\[
\alpha_2 \le n-1 < \alpha_1
\]
with strict inequality for $n \ge 2$. The following is a direct consequence of Theorem \ref{thm SER A new}, which cannot be obtained from Proposition 1.

\begin{cor}
\label{cor SER A new}
Consider the case of Theorem \ref{thm SER A new} when all decision regions are the spheres of same radius $d$. The following holds:
\begin{itemize}
\item The SER is strictly convex in $A$ in the large SNR regime:
 \[
P_{e}(A)'' > 0 \mbox{ if} \ A > \sqrt{n-1} /d
\]

\item It is strictly concave in the small SNR regime:
 \[
P_{e}(A)'' < 0 \mbox{ if} \ A < \sqrt{n-1} /d
\]

\item There is a single inflection point:
\[
P_{e}(A)'' = 0 \mbox{ if} \ A = \sqrt{n-1} /d
\]
\end{itemize}
\QED
\end{cor}

Theorem \ref{thm SER A new} can also be extended to a wide class of noise densities following the same approach as in Theorem \ref{thm SER SNR general}.

\begin{thm}
\label{thm SER A general}
Consider a decoder with center-convex decision regions operating in an additive noise channel of arbitrary density $f(r,\btheta)$, where $r$ represents the normalized noise amplitude $|\bxi|/\sigma_0$. The SER $P_e(A)$ is convex in $A$ in the interval $[A_1, A_2]$ if the noise density $f(r,\btheta)$ is non-increasing in $r$ in the interval $[A_1 d_{\min}, A_2 d_{\max}]$:
\begin{equation}
\label{eq3d-1}
P_e(A)'' \ge 0 \ \forall A \in [A_1, A_2] \ \mbox{if} \ f_r'(r,\btheta) \le 0 \ \forall \btheta, r \in [A_1 d_{\min}, A_2 d_{\max}],
\end{equation}
\end{thm}
\begin{proof}
See Appendix.
\end{proof}

Similarly to Theorem \ref{thm SER SNR general}, the inequalities can be reversed to obtain the concavity properties and unimodal densities are a special case. In particular, Corollary \ref{cor SER SNR unimodal} holds with the substitution $p \rightarrow r$, $\gamma \rightarrow A$, and \eqref{eq3a-8} reads as
\begin{align}
\label{eq3d-2}
\begin{cases}
P_e(A)'' > 0 \mbox{ if} \ A > r^*/d_{\min} \\
P_e(A)'' < 0 \mbox{ if} \ A < r^*/d_{\max} \\
\end{cases}
\end{align}
The Gaussian, Laplacian and exponential noise amplitude distributions are all unimodal, with $r^* = \sqrt{n-1}$, $r^* = n-1$ and $r^* = ((n-1)/2\beta)^{1/2\beta}$ respectively, so that the SER is always convex if $n=1$.

For the case of a SIRP noise as in Theorem \ref{thm SIRP SER convexity}, one obtains the following.
\begin{thm}
\label{thm SIRP conv A}
Consider an additive SIRP noise channel with the density as in \eqref{eq3b-1}, \eqref{eq3b-2} and $\bf{C} = I$. Assume that $f(\sigma)$ has bounded support: $f(\sigma)=0 \ \forall \sigma \notin [\sigma_1, \sigma_2]$. Then, the SER of any decoder with center-convex decision regions operating in this channel is convex at high and concave at low SNR as a function of signal amplitude $A$:
\begin{eqnarray}
P_e(A)'' &\ge& 0 \ \mbox{if} \ A \ge \sigma_2 \sqrt{n-1} / d_{\min} \\
P_e(A)'' &\le& 0 \ \mbox{if} \ A \le \sigma_1 \sqrt{n-1} / d_{\max}
\end{eqnarray}
where $d_{\min(\max)}$ is the minimum (maximum) distance of the normalized constellation (i.e. the one that corresponds to $A=1$).
\end{thm}
\begin{proof}
See Appendix.
\end{proof}

The following is immediate.

\begin{cor}
\label{cor T SIRP conv A}
Consider the scenario in Theorem \ref{thm SIRP conv A} for $n=1$ without the bounded support assumption. The SER is always convex in $A$: $P_e(A)'' \ge 0 \ \forall A$. \ \QED
\end{cor}

\subsection{Extension to Correlated Noise}
\label{s3d}
While Theorems 2, 5, 6, 8 and corresponding Corollaries apply to a channel with i.i.d. noise, a similar result can be established when noise is not i.i.d. (i.e. correlated or/and of non-identical variance per dimension). Let us consider the model in \eqref{eq2-1}, where the noise covariance is ${\bf{R}}_{\xi} = E\{\bxi \bxi^+\}$. Applying the sufficient statistics approach, one can use  ${\rm {\bf r}}' = {\bf{R}}_{\xi}^{-1/2} {\rm {\bf r}}$ instead of ${\rm {\bf r}}$ as decision variables without affecting the performance (i.e. a whitening filter). The equivalent channel
\begin{equation}
{\rm {\bf r}}'= {\bf{R}}_{\xi}^{-1/2}({\rm {\bf s}}+ \bxi)
\end{equation}
has i.i.d. noise ${\bf{R}}_{\xi}^{-1/2}\bxi$ and the equivalent constellation is $\{{\rm {\bf s}}_i'\} = \{{\bf{R}}_{\xi}^{-1/2}{\rm {\bf s}}_i\}$, so that equivalent decision regions $\Omega _i'$ and corresponding minimum/maximum distances can be found to which Theorems 2, 5, 6, 8 apply. In particular, the SER is still convex at high SNR. Note that Theorems 3 and 7 do not require the noise to be i.i.d.

\section{Convexity of BER and Capacity-Achieving Codes}
\label{sec4}

While the previous sections have established the convexity properties of the SER, it does not imply the corresponding convexity properties of the BER as the latter depends on the pairwise probability of error and not just the SER (see e.g. \eqref{eq2-9}). The PEP and the SER have somewhat different convexity properties. The convexity of the PEP has been established in \cite{Loyka-10} and, based on it, the following result was obtained.

\begin{thm}
\label{thm BER old}
Consider the ML decoder operating in the AWGN channel. Its BER, SER and PEP are all convex functions of the SNR, for any constellation, bit mapping and coding, in the high SNR (small noise) regime, when
\begin{equation}
\label{eq4-1}
d_{\min }^2 \ge (n+\sqrt {2n} ) \sigma_0^2,
\end{equation}
\end{thm}
\QED

Note that the lower bound in (\ref{eq4-1}) has an interesting interpretation: $n\sigma _0^2 $ is the mean of $|\bxi|^2$ and $\sqrt {2n} \sigma _0^2 $ is its standard deviation, so that (\ref{eq4-1}) requires that $d_{\min }^2 $ be larger than the average noise power by at least its standard deviation, which is intuitively what is required for low probability of error. Thus, the condition in (\ref{eq4-1}) should be satisfied when probability of error is small.

Below, we make this statement more precise and proceed to establish practical relevance of the high-SNR regime in (\ref{eq4-1}) based on the channel coding theorem. Recall that the sphere hardening argument (from the channel coding theorem) states that the noise vector $\bxi$ is contained within the sphere of radius $\sqrt {n\left( {\sigma _0^2 +\varepsilon } \right)} $ with high probability (approaching 1 as $n\to \infty )$ \cite{Shannon}\cite{Wozencraft}, where $\varepsilon >0$ is a fixed, arbitrary small number, so that the decision regions should have minimum distance to the boundary
\begin{equation}
d_{\min}^2 \ge n( {\sigma _0^2 +\varepsilon }),
\end{equation}
i.e. to enclose the hardened noise sphere of radius $\sqrt {n\left( {\sigma _0^2 +\varepsilon } \right)} $, to provide arbitrary low probability of error as $n\to \infty $. For any code satisfying this requirement, it follows that
\begin{equation}
\label{eq4-3}
d_{\min }^2 \ge n\left( {\sigma _0^2 +\varepsilon } \right)>(n+\sqrt {2n} )\sigma _0^2 ,
\end{equation}
for sufficiently large $n$ and $\forall \varepsilon >0$. Thus, for any code whose decision regions enclose the hardened noise spheres, the condition of Theorem \ref{thm BER old} is satisfied and therefore the error rates (SER, PEP, BER) of such codes are all convex.

On the other hand, for any code whose decisions regions are enclosed by the spheres of radius $\sqrt {n+\sqrt {2n} } \sigma _0$,  i.e. $d_{\max}^2 \le (n+\sqrt {2n} ) \sigma_0^2$ , the symbol error rates are lower bounded as
\begin{equation}
\label{eq4-4}
P_{ei} \ge \Pr \left\{ {\frac{|\bxi|-n}{\sqrt {2n} }>1} \right\}\approx Q(1)\approx 0.16>0,
\end{equation}
where $Q(x)=\textstyle{1 \over {\sqrt {2\pi } }}\int_x^\infty {e^{-t^2/2}dt} $ is the Q-function, so that arbitrary-low  probability of error is not achievable. Based on these two arguments, we conjecture the following.

\begin{con}
Consider a capacity-achieving code designed for $\mbox{SNR}=\gamma _0 $. Error rates of \textit{any} such code are convex for $\mbox{SNR}\ge \gamma _0 $, i.e. when it provides an arbitrary low probability of error. \QED
\end{con}

This conjecture is stronger that our convexity statement above since the latter requires the decision regions to include the hardened noise spheres, which is only a sufficient condition for arbitrarily low probability of error,  so that it is possible that a capacity-achieving code violates the condition in (\ref{eq4-3}).  The conjecture effectively states that, if present, such a violation is minor in nature and does not affect the convexity property.

As an application of this result, we note that power/time sharing cannot reduce error rates of any code for which (\ref{eq4-1}) holds. This complements the well-known result that power/time sharing cannot increase the capacity.

In summary, any code respecting the noise sphere hardening and hence having low probability of error will also have convex error rates (SER, PEP and BER). This is the way convexity intimately enters into the channel coding theorem.

Theorem \ref{thm BER old} can also be extended to a wide class of decoders with center-convex decision regions under a SIRP noise as follows. To separate the effects of noise power ($\sigma_0^2$) and the shape of its PDF, let us consider the normalized noise $\bxi/\sigma_0$  and assume that it has the PDF as in \eqref{eq3b-1}, \eqref{eq3b-2}, where $\sigma$ is now a normalized conditional standard deviation and $f(\sigma)$ is its PDF.

\begin{thm}
\label{thm BER SIRP}
Consider a channel with an additive SIRP noise as in Theorem \ref{thm SIRP representation} when the PDF $f(\sigma)$ of conditional normalized standard deviation has bounded support:
\begin{equation}
\label{eq4-5}
f(\sigma)=0 \ \forall \sigma \notin [\sigma_1, \sigma_2]
\end{equation}
The PEP, SER and BER of any decoder with center-convex decision regions (e.g. min-distance decoder for any constellation, bit mapping and coding) operating in this channel is a convex function of the SNR in the high SNR/low noise regime,
\begin{equation}
\label{eq4-6}
d_{\min }^2 \ge (n+\sqrt {2n} ) (\sigma_0 \sigma_2)^2
\end{equation}
\end{thm}
\begin{proof}
See Appendix.
\end{proof}

Theorem \ref{thm BER SIRP} essentially says that Theorem \ref{thm BER old} also applies to a SIRP noise channel provided the maximum conditional noise variance is used in \eqref{eq4-1}. We remark that it is only the constellation dimensionality and the minimum distance that determine its BER convexity and only via the bound in \eqref{eq4-6}, all its other details being irrelevant. As far as the noise is concerned, it is only the maximum conditional variance $(\sigma_0 \sigma_2)^2$ that matters and only via the same bound. A particular functional form of $f(\sigma)$ is irrelevant, i.e. many different unconditional noise distributions will induce the same convexity properties in the high-SNR regime.

\section{Convexity of Error Rates in Noise Power}
\label{sec5}
In a jammer optimization problem, it is convexity properties in noise power that are important \cite{Azizoglu}. Motivated by this fact, we study below convexity of the SER, the PEP and the BER in the noise power.

The following result has been established in \cite{Loyka-10}.

\begin{thm}[Theorem 4 in \cite{Loyka-10}]
\label{thm SER Pn old}
Consider the ML decoder operating in the AWGN channel. Its symbol error rates have the following convexity properties in the noise power $\sigma _0^2 $, for any constellation/coding,
\begin{itemize}
\item $P_{e} $ is concave in the large noise regime,
\begin{equation}
a_n \sigma _0^2 \ge d_{\max}^2
\end{equation}
where $a_n = n+2-\sqrt{2(n+2)}$,

\item $P_{e} $ is convex in the small noise regime,
\begin{equation}
b_n \sigma _0^2 \le d_{\min}^2
\end{equation}
where $b_n = n+2+\sqrt{2(n+2)}$,

\item there are an odd number of inflection points for intermediate noise power,
\begin{equation}
d_{\min}^2 b_n^{-1} \le \sigma _0^2 \le d_{\max}^2 a_n^{-1}
\end{equation}
\end{itemize}
\QED
\end{thm}

The following Theorem tightens the high/low SNR bounds above and also shows that the new bounds cannot be further improved in general.

\begin{thm}
\label{thm SER Pn new}
Consider a decoder with center-convex decision regions operating in the AWGN channel. Its symbol error rates have the following convexity properties in the noise power:
\begin{itemize}
\item $P_{e} $ is concave in the large noise regime,
\begin{equation}
\label{eq5-4}
(n+2) \sigma _0^2 \ge d_{\max}^2
\end{equation}

\item $P_{e} $ is convex in the small noise regime,
\begin{equation}
\label{eq5-5}
(n+2) \sigma _0^2 \le d_{\min}^2
\end{equation}

\item there are an odd number of inflection points for intermediate noise power,
\begin{equation}
d_{\min}^2 \le (n+2)\sigma _0^2 \le d_{\max}^2
\end{equation}

\item These bounds cannot be improved in the general case.
\end{itemize}
\end{thm}
\begin{proof}
See Appendix.
\end{proof}

Note that the bounds of Theorem \ref{thm SER Pn new} are indeed tighter than those of Theorem \ref{thm SER Pn old}, since
\begin{equation}
a_n < n+2 < b_n
\end{equation}
We further remark that similar results apply to $P_{e,i}$ via the substitution $d_{\max(\min)} \rightarrow d_{\max(\min),i}$.

The following Corollary, which cannot be obtained from Theorem \ref{thm SER Pn old}, follows immediately from Theorem \ref{thm SER Pn new}.

\begin{cor}
\label{cor SER Pn new}
Consider the scenario of Theorem \ref{thm SER Pn new} when all  decision regions are the spheres of same radius $d \ (= d_{\max} = d_{\min})$. The SER has the following convexity properties in noise power $\sigma_0^2$:
\begin{itemize}
\item $P_{e}(\sigma_0^2)$ is strictly concave in the large noise regime,
\begin{equation}
\sigma _0^2 > d^2/(n+2)
\end{equation}

\item It is strictly convex in the small noise regime,
\begin{equation}
\sigma _0^2 < d^2/(n+2)
\end{equation}

\item There is a single inflection point,
\begin{equation}
P_{e}(\sigma_0^2)''=0 \ \mbox{iff} \ \sigma _0^2 = d^2/(n+2)
\end{equation}
\end{itemize}
\QED
\end{cor}

Note that unlike $P_{e}(\gamma)$, which is convex in low dimensions ($n=1,2$) so that the transmitter cannot employ power/time sharing to reduce error rate, $P_{e}(\sigma_0^2)$ does not possess this property so that the jammer can increase error rate by power/time sharing even in low dimensions (in the low noise regime). In this respect, the jammer is in a more advantageous position compared to the transmitter in the AWGN channel. It is also clear from this Corollary that the high/low SNR bounds of Theorem \ref{thm SER Pn new} cannot be improved in general.

Theorem \ref{thm SER Pn new} can be also extended to a wider class of SIRP noise. As in Theorem \ref{thm BER SIRP}, we consider the normalized noise $\bxi/\sigma_0$ to separate the effects of the noise power ($\sigma_0^2$) and the shape of its PDF, and assume that the normalized noise power has the PDF as in \eqref{eq3b-1}, \eqref{eq3b-2}, where $f(\sigma)$ has bounded support as in \eqref{eq4-5}.

\begin{thm}
\label{thm SER Pn SIRP}
Consider a decoder with center-convex decision regions operating in a SIRP noise channel channel under the stated-above conditions. Its symbol error rates have the following convexity properties in the noise power:
\begin{itemize}
\item $P_{e} $ is concave in the large noise regime,
\begin{equation}
(n+2) (\sigma_1 \sigma _0)^2 \ge d_{\max}^2
\end{equation}

\item $P_{e} $ is convex in the small noise regime,
\begin{equation}
(n+2) (\sigma_2 \sigma _0)^2 \le d_{\min}^2
\end{equation}

\item there are an odd number of inflection points for intermediate noise power.
\end{itemize}
\end{thm}
\begin{proof}
See Appendix.
\end{proof}

Let us study now the convexity/concavity properties of the PEP as a function of noise power.

\begin{thm}
\label{thm PEP/BER conv Pn}
Consider a center-convex decoder operating in the AWGN channel. Its PEP $\Pr \{ {\bf s}_i \to {\bf s}_j \}$ is a convex function of the noise power $\sigma _0^2 $, for any $n$, in the low noise (high SNR) regime,
\begin{equation}
\label{eq5-10}
b_n \sigma _0^2 \le d_{\min ,i}^2
\end{equation}
and in the high noise (low SNR) regime,
\begin{equation}
\label{eq5-11}
a_n \sigma _0^2 \ge (d_{ij} +d_{\max,j} )^2
\end{equation}
where $d_{ij} = |{\bf s}_i - {\bf s}_j|$, and has an even number of inflection points in-between.
\end{thm}
\begin{proof}
See Appendix.
\end{proof}

Note that unlike the SER, the PEP is convex in the low SNR regime if $d_{\max,j} < \infty$. Based on this Theorem, a convexity property of the BER follows.

\begin{cor}
\label{cor BER conv Pn}
For any constellation, bit mapping and coding, the BER of a center-convex decoder operating in the AWGN channel is a convex function of the noise power in the low noise (high SNR) regime:
\begin{equation}
\label{eq5-12}
b_n \sigma _0^2 \le d_{\min }^2
\end{equation}
where the specifics of the constellation/code determine only the high-SNR boundary via $d_{\min }$. \ \QED
\end{cor}

We remark that for any code respecting the sphere hardening argument,
\begin{equation}
d_{\min }^2 \ge n \left( {\sigma _0^2 +\varepsilon } \right)> b_n \sigma _0^2 ,
\end{equation}
for sufficiently large $n$, so that the BER is a convex function of the noise power. For such codes, power/time sharing does not help to decrease the BER, but it is always helpful for a jammer whose objective is to increase the BER. A jammer transmission strategy to maximize the SER via a time/power sharing has been presented in \cite{Loyka-10} and, with some modifications, it can also be used to maximize the BER, following the convexity result in Corollary \ref{cor BER conv Pn}.

These results can also be extended to a SIRP noise channel.
\begin{cor}
\label{cor PEP/BER Pn SIRP}
Consider a SIRP noise channel, where  the conditional noise power $\sigma^2$ has bounded support
\begin{equation}
\notag
f(\sigma)=0 \ \forall \sigma \notin [\sigma_1, \sigma_2]
\end{equation}
The results of Theorem \ref{thm PEP/BER conv Pn} and Corollary \ref{cor BER conv Pn} apply with the substitutions  $\sigma_0 \rightarrow \sigma_2$ for \eqref{eq5-10} and \eqref{eq5-12}, and $\sigma_0 \rightarrow \sigma_1$ for \eqref{eq5-11}. \ \QED
\end{cor}

\section{Convexity in Fading Channels}
\label{sec 6}

The convexity properties of error rates in non-fading channels can also be extended to fading channels. Let us consider the following standard flat-fading channel model, which is a generalization of \eqref{eq2-1},
\begin{equation}
\label{eq6-1}
\br = h \bs+ \bxi
\end{equation}
where $h$ is a (scalar) fading channel gain, so that the instantaneous SNR is $\gamma = |h|^2 \gamma_0$, and the instantaneous error rate is $P_e(\gamma) = P_e(|h|^2 \gamma_0)$. The average error rate $\overline{P}_e(\gamma_0)$ as a function of the average SNR $\overline{\gamma}=\gamma_0 = 1/\sigma_0^2$ is obtained by the expectation over the fading distribution,
\begin{equation}
\label{eq6-2}
\overline{P}_e(\gamma_0) = \overline{P_e(\gamma)} = \int P_e(|h|^2 \gamma_0) f(h) dh
\end{equation}
where $f(h)$ is the PDF of $h$, and where $\overline{(\cdot)}$ denotes the expectation over the fading distribution.

If the instantaneous SER $P_e(\gamma)$ is convex for any SNR $\gamma$, the following result is immediate.

\begin{prop}
\label{prop fading no good}
Consider a fading channel under additive noise with monotonically-decreasing power density, e.g. a SIRP noise for $n \le 2$. The average SER of a decoder with center-convex decision regions operating in this channel is lower bounded by the non-fading SER at the same (average) SNR for \textit{any} fading distribution:
\begin{equation}
\overline{P}_e(\gamma_0) \ge P_e(\gamma_0)
\end{equation}
i.e. \textit{fading is never good in low dimensions under a SIRP noise}.
\end{prop}
\begin{proof}
Follows from Jensen inequality \cite{Boyd} by observing that $P_e(\gamma)$ is convex in $\gamma$ under the stated assumptions.
\end{proof}

Let us now consider the average error rate $\overline{P}_e(\gamma_0)$ as a function of the average SNR $\gamma_0$.

\begin{prop}
\label{prop SER fading}
Consider a fading channel in \eqref{eq6-1} and assume that the instantaneous SER $P_e(\gamma)$ is convex for any SNR $\gamma$ (e.g. a SIRP noise for $n \le 2$ or any noise with monotonically-deceasing power density under a center-convex decoder), then the average SER $\overline{P}_e(\gamma_0)$ is also convex in the average SNR $\gamma_0$ in such channel, i.e. \textit{flat-fading is a convexity-preserving process}.
\end{prop}
\begin{proof}
Follows from \eqref{eq6-2} since non-negative linear combination preserves convexity \cite{Boyd} or, equivalently, by using the convexity-preserving property in \eqref{eq3c-7}.
\end{proof}

We note that Propositions \ref{prop fading no good} and \ref{prop SER fading} extend the corresponding results in \cite{Loyka-10} obtained for the Gaussian noise and the ML decoder to a broad class of noise distributions and decoders. It appears that it is the constellation dimensionality that has a major impact on convexity of the SER, rather than the specifics of the noise or the fading distribution.

These results can be further extended to diversity combining systems over such channel, which is a popular way to combat the detrimental effects of fading \cite{Wozencraft}-\cite{Benedetto}.

\subsection{Convexity under diversity combining}
\label{sec6a}

Consider a maximum ratio combiner (MRC) operating over an $m$-branch fading channel as in \eqref{eq6-1},
\begin{equation}
\label{eq6-4}
\br_i = h_i \bs+ \bxi_i
\end{equation}
where $\br_i$, $h_i$ and $\bxi_i$, are the received signal, channel (voltage) gain and noise in $i$-th branch, $i=1..m$. The $i$-th branch SNR is $\gamma_i = |h_i|^2 \gamma_0$ and the combiner's output SNR is $\gamma_{out} = \sum_i \gamma_i = \gamma_0 |\bh|^2$ \cite{Barry}, where $\bh = [h_1,..,h_m]^T$ is the vector of channel gains. The combiner's instantaneous error rate is $P_e(|\bh|^2 \gamma_0)$ and the average error rate is
\begin{equation}
\label{eq6-5}
\overline{P}_e(\gamma_0) = \int P_e(|\bh|^2 \gamma_0) f(\bh) d\bh
\end{equation}
Using the same argument as in Proposition \ref{prop SER fading}, this error rate is convex in the average SNR $\gamma_0$ provided the instantaneous SER is convex for any SNR.

This result can be now extended to an arbitrary linear combining of the form $\sum_i w_i \br_i$, where $\bw = [w_1,...,w_m]^T$ are the combining weights (which depend on the channel gains). The output SNR of this combiner is $\gamma_{out} = \gamma_0 |\bw^T \bh|^2$, assuming proper normalization $|\bw|=1$ (note that normalization does not affect the SNR), so that its average error rate is
\begin{equation}
\label{eq6-6}
\overline{P}_e(\gamma_0) = \int P_e(|\bw^T\bh|^2 \gamma_0) f(\bh) d\bh
\end{equation}
which is also convex in the average SNR provided $P_e(\gamma)$ is convex, i.e. any linear combing is convexity-preserving. Note that the MRC is a special case of the general linear combining, with $\bw = \bh/|\bh|$. Other special cases are the other 2 popular combining techniques: selection combining (SC), which selects the strongest branch with only one non-zero weight corresponding to that brach, and equal-gain combing, which adds coherently the required signals with unit gain \cite{Barry}. All of them preserve the convexity of error rates, which is summarized below.
\begin{prop}
\label{prop SER combining}
\textit{Any} linear diversity combining over \textit{any} flat-fading channel as in \eqref{eq6-4} is convexity-preserving, i.e. given that the instantaneous SER $P_e(\gamma)$ is convex for any SNR $\gamma$, the average SER $\overline{P}_e(\gamma_0)$ of the combiner is also convex in the average SNR $\gamma_0$ in such channel. Special cases include the maximum ratio, selection and equal gain combining. The lower bound in Proposition \ref{prop fading no good} also holds under any linear combining.
\end{prop}

The utility of this convexity-preserving property is coming from the fact that most error rate expressions in fading channels and under diversity combining are prohibitively complex so that the straightforward evaluation of convexity via differentiation is not possible, while the results above establish the convexity indirectly and without evaluating the integrals (the most difficult part).

\section{Conclusion}

Convexity/concavity properties of the error rates (SER, PEP, BER) in an additive noise channel have been considered. The earlier results obtained for the AWGN channel under ML (min-distance) decoder \cite{Loyka-10} have been improved and have also been extended to a class of decoders with center-convex decision regions and to a wide class of noise densities (unimodal and SIRP noise processes). In particular, the SER is shown to be a convex function of the SNR for any noise with monotonically-decreasing power density (e.g. SIRP, Laplacian, Weibull, power-exponential or AWGN noise in low dimensions). In higher dimensions, this property holds in the high SNR regime, for which the boundary has been explicitly given. The latter is such that any code that respects the sphere hardening condition of the channel coding Theorem also meets the high SNR condition so that all such codes have convex error rates (SER, PEP, BER). Fading is shown to be a convexity-preserving process, including any linear combining, and is never good in low dimensions under a SIRP noise.

All the applications discussed earlier in \cite{Loyka-10} (e.g. optimization of a spatial multiplexing system, optimum power/time sharing for a jammer and transmitter, optimal unitary precoding for an OFDM system) for the AWGN channel also hold under the general SIRP or unimodal noise and a convex-center decoder, based on the convexity properties established here.

\section{Appendix}
\label{sec Appendix}

\subsection{Proof of Theorem \ref{thm SER SNR new}}

First, we transform the Cartesian noise density $f_{\xi}(\bf{x})$ into the spherical coordinates $(p,\btheta)$, where $p$ represents the normalized noise instant power $|\bf{x}|^2/\sigma_0^2$, and $\btheta = \{\theta_1,..,\theta_{n-1}\}$ are the angles (see \cite{Anderson}\cite{Goldman} for more on spherical coordinates):
\begin{align}
\label{eqAa-1}
f_{p,\btheta}(p,\btheta) &= f_{\btheta}(\btheta) f_p(p), \\ \notag
-\pi/2 \le \theta_i &\le \pi/2, i=1...n-2, -\pi \le \theta_{n-1} \le \pi
\end{align}
where $f_{\btheta}(\btheta)$ and $f_p(p)$ are the angular and normalized noise power densities,
\begin{eqnarray}
\label{eqAa-2}
f_{\btheta}(\btheta) &=& \Gamma(n/2) \pi^{-n/2} \cos^{n-2}\theta_{n-2} \cos^{n-3}\theta_{n-3} .. \cos \theta_1 \\
\label{eqAa-3}
f_p(p) &=& \frac{p^{n/2-1} e^{-p/2}}{2^{n/2} \Gamma(n/2)}
\end{eqnarray}
where $\Gamma(\cdot)$ is Gamma function. Using this, the probability of correct decision $P_{ci}$ can be expressed as
\begin{equation}
\label{eqAa-4}
P_{ci}(\gamma) = \int_{D\btheta} f_{\btheta}(\btheta) \int_{0}^{\gamma R_i^2(\btheta)} f_p(p) d{p} d\btheta
\end{equation}
where $D \btheta$ is the range of angles in \eqref{eqAa-1} and $R_i(\btheta)$ is the boundary of the normalized decision region (corresponding to $\gamma=1/\sigma_0^2=1$). One can now obtain the second derivative in $\gamma$:
\begin{equation}
\label{eqAa-5}
P_{ci}(\gamma)'' = \int_{D\btheta} f_{\btheta}(\btheta) f_p'(\gamma R_i^2(\btheta))  R_i^4(\btheta) d\btheta
\end{equation}
where
\begin{equation}
\label{eqAa-6}
f_p'(p) = \frac{(n-2-p) p^{n/2-2} e^{-p/2}}{2^{n/2+1} \Gamma(n/2)}
\end{equation}
so that $f_p'(\gamma R_i^2(\btheta)) \le 0$ if $\gamma R_i^2(\btheta) \ge n-2$. When the latter condition holds for any $\btheta$, i.e. when $\gamma \ge (n-2)/d_{\min,i}^2$, then the integrand in \eqref{eqAa-5} is non-positive, since $f_{\btheta}(\btheta)R_i^4(\btheta) \ge 0$, and \eqref{eq3a-1} follows. Reversing the inequalities, one obtains \eqref{eq3a-2}. To prove that the bound in \eqref{eq3a-1} cannot be further improved in general (i.e. without further assumptions on the constellation geometry), consider the case when all $\Omega_i$ are spheres of the same radius $r = d_{\min}=d_{\max} < \sqrt{(n-2)/\gamma}$, so that
$\gamma < (n-2)/d_{\min}^2=(n-2)/d_{\max}^2$ and therefore  $P_e(\gamma)'' < 0$, i.e. \eqref{eq3a-1} is necessary for the convexity of the SER in general. The bound in \eqref{eq3a-2} can be handled in the same way. The case of identical spherical decision regions follows in a straightforward way. As a side remark, we note that this proof is a significant simplification over those of Theorems 1 and 2 in \cite{Loyka-10}. \ \QED

\subsection{Proof of Theorem \ref{thm SER SNR general}}
\label{sec Ab}

In the case of generic noise density $f(p,\btheta)$, i.e. when \eqref{eqAa-1} does not hold, \eqref{eqAa-4} and \eqref{eqAa-5} are generalized to
\begin{equation}
\label{eqAb-1}
P_{ci}(\gamma) = \int_{D\btheta} \int_{0}^{\gamma R_i^2(\btheta)} f(p,\btheta) d{p} d\btheta
\end{equation}
\begin{equation}
\label{eqAb-2}
P_{ci}(\gamma)'' = \int_{D\btheta} f_p'(\gamma R_i^2(\btheta),\btheta)  R_i^4(\btheta) d\btheta
\end{equation}
Now observe that $P_{ci}(\gamma)'' \le 0$ if $f_p'(\gamma R_i^2(\btheta),\btheta) \le 0$ $\forall \btheta$, which holds if $f_p'(p,\btheta) \le 0$ for $\gamma d_{\min}^2 \le p \le \gamma d_{\max}^2,$ and all $\btheta$, so that \eqref{eq3a-5} follows. \eqref{eq3a-6} follows by observing that its condition insures that the condition in \eqref{eq3a-5} is satisfied for $\gamma_1 \le \gamma \le \gamma_2$. \ \QED

\subsection{Proof of Theorem \ref{thm SIRP SER convexity}}
\label{sec Ac}

Using Theorem \ref{thm SIRP representation}, the noise power density can be written as
\begin{equation}
\label{eqAc-1}
f(p) = \int_0^{\infty} f(p|\sigma) f(\sigma) d\sigma
\end{equation}
where $f(p|\sigma) $ is the conditional power density,
\begin{equation}
\label{eqAc-2}
f(p|\sigma) =  \frac{1}{2^{n/2} \Gamma(n/2)\sigma^2} \left(\frac{p}{\sigma^2}\right)^{n/2-1} \exp \left\{-\frac{p}{2\sigma^2}\right\}
\end{equation}
\eqref{eqAb-1} and \eqref{eqAb-2} can be written as
\begin{equation}
\label{eqAc-3}
P_{ci}(p_s) = \int_{D\btheta} \int_{\sigma_1}^{\sigma_2} f_{\btheta}(\btheta) f(\sigma) \int_{0}^{p_s R_i^2(\btheta)} f(p|\sigma)d{p} d\sigma d\btheta
\end{equation}
\begin{equation}
\label{eqAc-4}
P_{ci}(p_s)'' = \int_{D\btheta} \int_{\sigma_1}^{\sigma_2} f_{\btheta}(\btheta) f(\sigma) f_p'(p_s R_i^2(\btheta)|\sigma) R_i^4(\btheta) d\sigma d\btheta
\end{equation}
where now $R_i(\btheta)$ is the boundary of normalized decision region corresponding to $p_s=1$, and $f_{\btheta}(\btheta)$ is as in \eqref{eqAa-2}. Observe from \eqref{eqAc-2} that
\begin{eqnarray}
\label{eqAc-5}
f_p'(p|\sigma) &\ge& 0 \ \mbox{if} \ p \le (n-2)\sigma^2 \\
\label{eqAc-6}
f_p'(p|\sigma) &\le& 0 \ \mbox{if} \ p \ge (n-2)\sigma^2
\end{eqnarray}
i.e. $f(p|\sigma)$ is unimodal in $p$ with $p^* = (n-2)\sigma^2$.  Using these properties in \eqref{eqAc-4} and observing that $d_{\min,i} \le R_i(\btheta) \le d_{\max,i}$, $\sigma_1 \le \sigma \le \sigma_2$, the integrand in \eqref{eqAc-4} is non-negative/non-positive if
\begin{eqnarray}
\label{eqAc-7}
p_s d_{\max,i}^2 &\le (n-2)\sigma_1^2 \\
p_s d_{\min,i}^2 &\ge (n-2)\sigma_2^2
\end{eqnarray}
so that $P_{ci}(p_s)'' \ge 0$ or $P_{ci}(p_s)'' \le 0$ from which \eqref{eq3b-3} and \eqref{eq3b-4} follow.

\subsection{Proof of Theorem \ref{thm SER A new}}
\label{sec Ad}
Using \eqref{eqAa-2}-\eqref{eqAa-5}, $P_{ci}(A)$ can be written as
\begin{equation}
P_{ci}(A) = \int_{D\btheta} f_{\btheta}(\btheta)   \int_{0}^{A R_i(\btheta)} f(r) d{r}d\btheta
\end{equation}
where $f_r(r)$ is the normalized noise amplitude density,
\begin{eqnarray}
f(r) &=& \frac{r^{n-1} e^{-r^2/2}}{2^{n/2-1} \Gamma(n/2)}
\end{eqnarray}
and $R_i(\btheta)$ is the decision region boundary of the normalized ($A=1$) constellation. Therefore,
\begin{equation}
\label{eqAd-3}
P_{ci}(A)'' = \int_{D\btheta} f_{\btheta}(\btheta) f_r'(A R_i(\btheta))  R_i^2(\btheta) d\btheta
\end{equation}
where
\[
f_r'(r) = \frac{r^{n-2} e^{-r^2/2}(n-1-r^2)}{2^{n/2-1}\Gamma(n/2)}
\]
so that the integrand in \eqref{eqAd-3} is non-negative when $A R_i(\btheta) \le \sqrt{n-1} \ \forall \btheta$, which is case when $A d_{\max,i} \le \sqrt{n-1}$ so that $P_{ei}(A)'' \le 0$. The opposite case is similar. An odd number of inflection points follows from the continuity argument. The fact that the high/low SNR bounds cannot be tightened in general is clear from Corollary \ref{cor SER A new}.

\subsection{Proof of Theorem \ref{thm SER A general}}
\label{sec Ae}
In this case, \eqref{eqAb-1} and \eqref{eqAb-2} become
\begin{equation}
\label{eqAe-1}
P_{ci}(A) = \int_{D\btheta} \int_{0}^{A R_i(\btheta)} f(r,\btheta) d{r} d\btheta
\end{equation}
\begin{equation}
\label{eqAe-2}
P_{ci}(A)'' = \int_{D\btheta} f_r'(A R_i(\btheta),\btheta)  R_i^2(\btheta) d\btheta
\end{equation}
and the argument of section \ref{sec Ab} goes thorough with the substitution $\gamma \rightarrow A$, $p \rightarrow r$.

\subsection{Proof of Theorem \ref{thm SIRP conv A}}
\label{sec Af}
In this case, \eqref{eqAc-1}-\eqref{eqAc-4} are modified to
\begin{equation}
\label{eqAf-1}
f(r) = \int_0^{\infty} f(r|\sigma) f(\sigma) d\sigma
\end{equation}
\begin{equation}
\label{eqAf-2}
f(r|\sigma) =  \frac{1}{2^{n/2-1} \Gamma(n/2)\sigma} \left(\frac{r}{\sigma}\right)^{n-1} \exp \left\{-\frac{r^2}{2\sigma^2}\right\}
\end{equation}
\begin{equation}
\label{eqAf-3}
P_{ci}(A) = \int_{D\btheta} \int_{\sigma_1}^{\sigma_2} f_{\btheta}(\btheta) f(\sigma) \int_{0}^{A R_i(\btheta)} f(r|\sigma)d{r} d\sigma d\btheta
\end{equation}
\begin{equation}
\label{eqAf-4}
P_{ci}(A)'' = \int_{D\btheta} \int_{\sigma_1}^{\sigma_2} f_{\btheta}(\btheta) f(\sigma) f_r'(A R_i(\btheta)|\sigma) R_i^2(\btheta) d\sigma d\btheta
\end{equation}
and the rest of the proof in section \ref{sec Ac} goes through with appropriate modifications.

\subsection{Proof of Theorem \ref{thm BER SIRP}}
\label{sec Ag}
The SIRP noise Cartesian PDF can be expressed as:
\begin{equation}
f(\bx) = \frac{1}{(2\pi)^{n/2}} \int_{\sigma_1}^{\sigma_2}\sigma^{-n} g(\gamma,\sigma,\bx) f(\sigma) d\sigma
\end{equation}
where
\begin{equation}
g(\gamma,\sigma,\bx) =  \gamma^{n/2} \exp \left\{-\frac{\gamma |\bx|^2}{2 \sigma^2}\right\}
\end{equation}
and the SNR $\gamma = 1/\sigma_0^2$, so that
\begin{equation}
\Pr \{ {\bf s}_i \to {\bf s}_j \}''_{\gamma} =\int_{\Omega_j} f(\bx)''_{\gamma} \ d\bx
\end{equation}
where $\Omega_j$ is the decision region for ${\bf s}_j$ while the reference frame is centered on ${\bf s}_i$. Now observe that
\begin{equation}
g(\gamma,\sigma,\bx)''_{\gamma} = \frac{1}{4} \gamma^{n/2} (w-w_1)(w-w_2) \exp (-\gamma w/2)
\end{equation}
where $w=|\bx|^2 \sigma^{-2}$ and $w_{1(2)} = (n \pm \sqrt{2n})/\gamma$, so that
\begin{equation}
g(\gamma,\sigma,\bx)''_{\gamma} \ge  0 \ \mbox{if} \ w \ge w_1
\end{equation}
and hence
\begin{equation}
\label{eqAg-6}
f(\bx)''_{\gamma} = \frac{1}{(2\pi)^{n/2}} \int_{\sigma_1}^{\sigma_2}\sigma^{-n} g(\gamma,\sigma,\bx)''_{\gamma} f(\sigma) d\sigma \ge 0 \ \forall \bx \in \Omega_j \ \mbox{if} \ d_{\min,i}^2 \ge w_1 \sigma_2^2
\end{equation}
from which it follows that
\begin{equation}
\Pr \{ {\bf s}_i \to {\bf s}_j \}''_{\gamma} \ge 0 \ \forall i,j \ \mbox{if} \ d_{\min}^2 \ge w_1 \sigma_2^2
\end{equation}
and, from \eqref{eq2-9}, $\textsf{BER}''_{\gamma} \ge 0$ under the same condition, so that \eqref{eq4-6} follows, where $\gamma = 1/\sigma_0^2$.

\subsection{Proof of Theorem \ref{thm SER Pn new}}
\label{sec Ah}
The proof follows along the same lines as that of Theorem \ref{thm SER SNR new}, with the substitution $\gamma = 1/\sigma_0^2$. In particular, \eqref{eqAa-5} is modified to
\begin{equation}
\label{eqAh-1}
P_{ci}(\sigma_0^2)'' = \sigma_0^{-6} \int_{D\btheta} f_{\btheta}(\btheta) \{f_p'(\sigma_0^{-2} R_i^2(\btheta))  R_i^2(\btheta) + 2 \sigma_0^2 f_p(\sigma_0^{-2} R_i^2(\btheta))\} d\btheta
\end{equation}
where the derivative of the noise power density $f_p'(p)$ is as in \eqref{eqAa-6}. The integrand in \eqref{eqAh-1} is non-negative when $d_{\max,i}^2 \le (n+2)\sigma_0^2$ and non-positive when $d_{\min,i}^2 \ge (n+2)\sigma_0^2$ from which \eqref{eq5-4} and \eqref{eq5-5} follow. The inflection points follow from the continuity argument. The fact that the bounds cannot be improved is clear from the equal spherical decision regions of Corollary \ref{cor SER Pn new}.

\subsection{Proof of Theorem \ref{thm SER Pn SIRP}}
\label{sec Ai}
The proof is essentially a generalized version of the previous proof. Under the stated conditions, $P_{ci}(\sigma_0^2)$ can be written as
\begin{equation}
\label{eqAi-1}
P_{ci}(\sigma_0^2) = \int_{D\btheta} f_{\btheta}(\btheta) \int_{\sigma_1}^{\sigma_2} f(\sigma) \int_{0}^{\sigma_0^{-2} R_i^2(\btheta)} f(p|\sigma)d{p} d\sigma d\btheta
\end{equation}
where $R_i(\btheta)$ is the boundary of normalized decision region corresponding to $\sigma_0=1$ and $f(p|\sigma)$ is as in \eqref{eqAc-2}, so that its second derivative in $\sigma_0^2$ is
\begin{equation}
\label{eqAi-2}
P_{ci}(\sigma_0^2)'' = \sigma_0^{-6} \int_{D\btheta} f_{\btheta}(\btheta) \int_{\sigma_1}^{\sigma_2}  f(\sigma) \{f_p'(\sigma_0^{-2} R_i^2(\btheta)|\sigma)  R_i^2(\btheta) + 2 \sigma_0^2 f_p(\sigma_0^{-2} R_i^2(\btheta)|\sigma)\} d\sigma d\btheta
\end{equation}
and the integrand in \eqref{eqAi-2} is non-negative when $d_{\max,i}^2 \le (n+2)(\sigma_1 \sigma_0)^2$ and non-positive when $d_{\min,i}^2 \ge (n+2)(\sigma_2 \sigma_0)^2$ from which the Theorem follows.

\subsection{Proof of Theorem \ref{thm PEP/BER conv Pn}}
\label{sec Aj}

The PEP can be expressed as
\begin{equation}
\Pr \{ {\bf s}_i \to {\bf s}_j \} =\int_{\Omega_j} f(\bx)d\bx
\end{equation}
where  $\Omega_j$ is the decision region for ${\bf s}_j$ while the reference frame is centered on ${\bf s}_i$, and the AWGN density $f(\bx)$ is as in \eqref{eq2-2}. The second derivative in the noise power $\sigma_0^2$ is
\begin{equation}
\Pr \{ {\bf s}_i \to {\bf s}_j \}''_{\sigma_0^2} =\int_{\Omega_j} f(\bx)''_{\sigma_0^2} \ d\bx
\end{equation}
where $f(\bx)''_{\sigma_0^2}$ can be expressed as
\begin{equation}
f(\bx)''_{\sigma_0^2} = \frac{\sigma_0^8}{4} \left( \frac{\sigma_0^2}{2\pi}\right)^{n/2} \exp \left\{-\frac{|\bx|^2}{2 \sigma_0^2}\right\} g(|\bx|^2)
\end{equation}
and
\begin{equation}
g(x) = (x - b_n \sigma_0^2)(x - a_n \sigma_0^2)
\end{equation}
where $a_n, b_n$ are as in Theorem \ref{thm SER Pn old}. Clearly, $g(x) \ge 0$ if $x \ge b_n \sigma_0^2$ or $x \le a_n \sigma_0^2$, so that $f(\bx)''_{\sigma_0^2} \ge 0 \ \forall \bx \in \Omega _j$ if $d_{\min,i}^2 \ge b_n \sigma_0^2$ or $a_n \sigma _0^2 \ge (d_{ij} +d_{\max,j} )^2$ and the result follows. Figures 5 and 6 illustrate these two cases.

\begin{figure}[t]
\label{fig_5}
\centerline{\includegraphics[width=3.in]{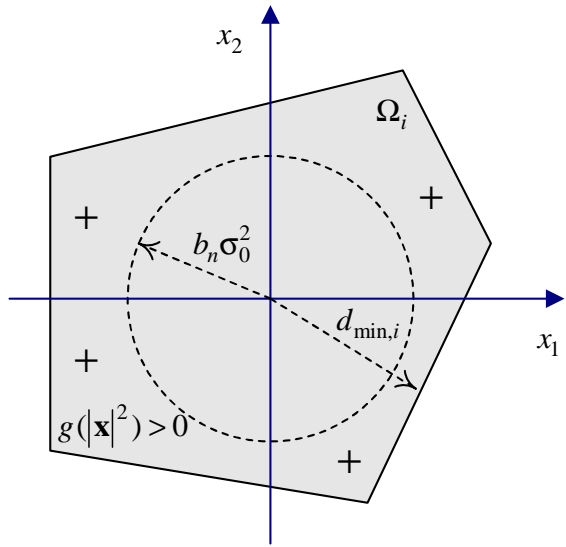}}
\caption{Two-dimensional illustration of the problem geometry for the case $d_{\min,i}^2 \ge b_n \sigma_0^2$. The decision region $\Omega _i$ is shaded. $g(|\bx|^2)$ has a sign as indicated by ``+'' and ``-``.}
\end{figure}

\begin{figure}[t]
\label{fig_6}
\centerline{\includegraphics[width=3.in]{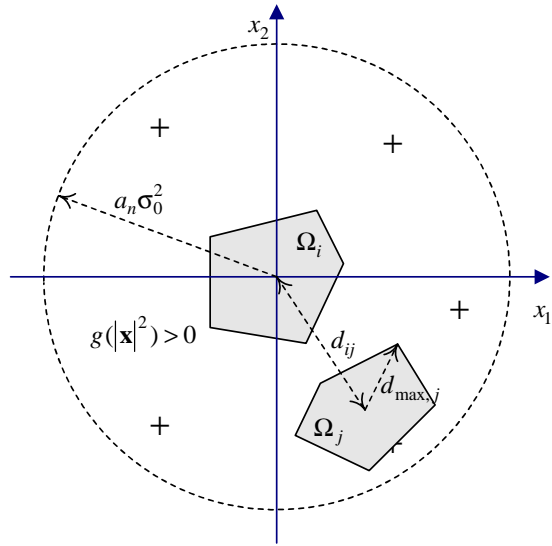}}
\caption{Two-dimensional illustration of the problem geometry for the case $a_n \sigma _0^2 \ge (d_{ij} +d_{\max,j} )^2$.}
\end{figure}

\end{document}